\documentclass[twoside,leqno,11pt]{article}  
\usepackage{ltexpprt} 

\usepackage{times}
\usepackage{xspace}
\usepackage{subfigure}
\usepackage{pdflscape}
\usepackage{xcolor}
\usepackage[font=small]{caption}
\usepackage{hyperref}
\usepackage{algorithmic}
\usepackage{algorithm}
\usepackage{color, colortbl}
\usepackage{enumerate}
\usepackage{amsmath}
\usepackage{multirow}
\usepackage{mdwlist}
\usepackage{hhline}
\usepackage{graphicx}
\usepackage{amsfonts}

\definecolor{Gray}{gray}{0.7}

\global\long\def\phihalf{{\vec{s_{0}}}}

\global\long\def\y{{\vec{y}}}
\newcommand{\bp}{\textsc{BP}\xspace}
\newcommand{\rwr}{\textsc{RWR}\xspace}
\newcommand{\fabp}{\textsc{FaBP}\xspace}
\global\long\def\matA{\mathbf{A}}
\global\long\def\matI{\mathbf{I}}
\global\long\def\matD{\mathbf{D}}
\global\long\def\matB{\mathbf{S}}
\global\long\def\matF{\mathbf{F}}
\global\long\def\matL{\mathbf{L}}
\global\long\def\unitv{\vec{e}}
\global\long\def\cBP{c'}
\global\long\def\ccBP{\epsilon}
\global\long\def\aa{\epsilon^2}

\newcommand{\bhalf}{ {\vec{s} }}
\global\long\def\hhalf{h_{h}}
\newcommand{\methodSlow}{\textsc{DeltaCon$_{0}$}\xspace}
\newcommand{\method}{\textsc{DeltaCon}\xspace}

\newcommand{\methodT}{\textbf{{\large \textsc{DeltaCon}}}\xspace}

\newcommand{\myscshape}[1]{\textbf{\scshape{#1}}}
\newcommand{\mypar}[1]{\textbf{#1.}}

\newcolumntype{y}{>{\columncolor{yellow!90!white}}c}
\newcommand{\wrong}[1]{\multicolumn{1}{y|}{#1}}
\newcommand{\nodeSet}{\mathcal{V}}
\newcommand{\edgeSet}{\mathcal{E}}
\newcommand{\n}{n}
\newcommand{\m}{m}

\newcommand{\hide}[1]{}
\newcommand{\bit}{\begin{itemize}}
\newcommand{\eit}{\end{itemize}}
\newcommand{\reminder}[1]{ {\color{red} {\bf *** #1} } }

\newcommand{\mkclean}{
    \renewcommand{\reminder}[1]{}
}

\newcommand{\matusita}{\textsc{RootED}\xspace}

  \newtheorem{observation}{Observation}

  \newtheorem{problem}{Problem}
  \newtheorem{thm}{Theorem}
  \newtheorem{intuition}{Intuition}

\mkclean
\title{\method: A Principled Massive-Graph Similarity Function}

\begin{document}
%
%
%
\author{Danai Koutra\thanks{Computer Science Department, Carnegie Mellon University.} \\danai@cs.cmu.edu
\and 
Joshua T. Vogelstein\thanks{Department of Statistical Science, Duke University.}\\jovo@stat.duke.edu
\and
Christos Faloutsos\footnotemark[1]\\christos@cs.cmu.edu
}
\date{}

%
\newenvironment{oneshot}[1]{\@begintheorem{#1}{\unskip}}{\@endtheorem}

\maketitle
%
%
%
%
\begin{abstract}

How much did a network change since yesterday? 
How different is the wiring between Bob's brain (a left-handed male)
and Alice's brain (a right-handed female)?
Graph similarity with known node correspondence, i.e. the detection of changes in the connectivity of graphs, arises in numerous settings.
In this work, we formally state the axioms and desired properties of the graph similarity functions, and evaluate when state-of-the-art methods fail to detect
crucial connectivity changes in graphs. We propose \method, a principled, intuitive, and scalable algorithm that assesses the similarity between two graphs
on the same nodes (e.g. employees of a company, customers of a mobile carrier). 
Experiments on various synthetic and real graphs showcase the advantages of our method over existing similarity measures. Finally, we employ \method
to real applications: (a) we classify people to groups of high and low creativity based on their brain connectivity graphs, and (b) do temporal anomaly detection in the who-emails-whom Enron graph. 


\hide{
\reminder{--BP/diffusion method                     
- the proposed method can detect subtle changes in small and large graphs
- scalable
- increasing the number of different edges between the compared graphs leads to decrease of the similarity measure (as expected)
- weights of edges are taken into account and are penalized more than other missing edges
-targeted vs random change}
}

\end{abstract}
\section{Introduction}
Graphs arise naturally in numerous situations; social, traffic, collaboration and computer networks, images, protein-protein interaction networks, brain connectivity graphs and web graphs are only a few examples. 
A problem that comes up often in all those settings is the following: how much do two graphs or networks differ in terms of connectivity?

Graph similarity (or comparison) is a core task for sense-making:
 abnormal changes in the network traffic may indicate a computer attack; 
differences of big extent in a who-calls-whom graph may reveal a national celebration, 
or a telecommunication problem. Besides, network similarity can 
give insights into behavioral patterns: is the Facebook message graph similar to the Facebook wall-to-wall graph? 
Tracking changes in networks over time, spotting anomalies and detecting events is a research direction that has attracted much interest (e.g., \cite{CaceresBG11}, \cite{NobleC03}, \cite{WangPT11}). 
%
%

\begin{figure}
\centering
\subfigure[Connectome: neural network of brain.]{\includegraphics[width=0.3\columnwidth]{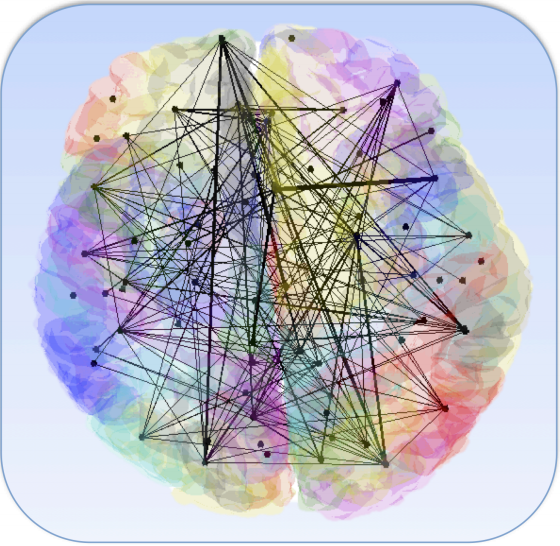}}
\subfigure[Dendogram representing the hierarchical clustering of the \method similarities between the 114 connectomes.]{\includegraphics[width=0.4\columnwidth]{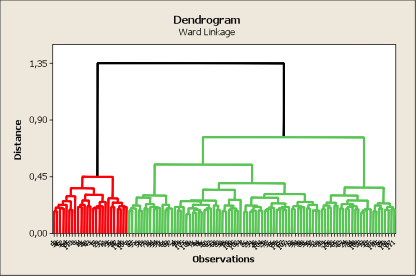}\label{fig:clustering}}
\caption{(a) Brain network (connectome). Different colors correspond to each of the 70 cortical regions, whose centers are depicted by vertices. Connections between the regions are shown by edges.
\method is used for clustering and classification.
(b) The connectomes are nicely classified in two big clusters by hierarchical clustering. The classification is based on the pairwise \method similarities between the 114 connectomes that we study. Elements in red correspond to high artistic score - thus, \method shows that artistic brains seem to have different wiring than the rest. }
\label{fig:brain}
\end{figure}

Long in the purview of researchers, graph similarity is a well-studied problem and several approaches have been proposed to solve variations of the problem. However, graph comparison still remains an open problem, while, with the passage of time, the list of requirements increases: the exponential growth of graphs, both in number and size, calls for methods that are not only accurate, but also scalable to graphs with billions of nodes.


In this paper, we address two main questions: How to compare two networks efficiently? How to evaluate their similarity score? 
Our main contributions are the following:
\begin{enumerate*}
\vspace{-0.2cm}
\item \emph{Axioms/Properties}: we formalize the axioms and properties that a similarity measure must conform to.
\item \emph{Algorithm}: we propose  \method for measuring connectivity differences between two graphs, and show that it is: (a)  \emph{principled}, conforming to all the axioms presented in Section~\ref{sec:proposed}, (b) \emph{intuitive}, giving similarity scores that agree with common sense and can be easily explained, and (c) \emph{scalable}, able to handle large-scale graphs.
\item \emph{Experiments}: we report experiments on synthetic and real datasets, and compare \method to six state-of-the-art methods that apply to our setting. 
\item \emph{Applications}: We use \method for real-world applications, such as temporal anomaly detection and clustering/classification. In Fig.~\ref{fig:brain}, \method is used for clustering brain graphs corresponding to 114 individuals;  
the two big clusters which differ in terms of connectivity correspond to people with high and low creativity. More details are given in Sec.~\ref{sec:applications}.
\end{enumerate*}

%
The paper is organized as follows: Section \ref{sec:proposed} presents the intuition behind our method, and the axioms and desired properties of a similarity measure; Sec.~\ref{sec:proposedDetails} has the proposed algorithms; experiments on synthetic and big real networks are in Sec.~\ref{sec:experiments}; Sec.~\ref{sec:applications} presents two real-world applications; the related work and the conclusions are in Sec.~\ref{sec:related} and ~\ref{sec:conclusions} respectively. 
Finally, Table \ref{tab:definitions} presents the major symbols we use in the paper and their definitions.

\section{Proposed Method: Intuition}
\label{sec:proposed}
\definecolor{Gray}{gray}{0.9}
\newcolumntype{g}{>{\columncolor{Gray}}c}
\begin{table}
\centering
\caption{Symbols and Definitions. Bold capital letters: matrices, lowercase letters with arrows: vectors, plain font: scalars.}
\label{tab:definitions}
\small{
\begin{tabular}{|c||l|} \hline
\multicolumn{1}{|g||}{\textbf{Symbol}} & \multicolumn{1}{g|}{\textbf{Description}} \\ \hline \hline
\multicolumn{1}{|g|}{$G$} & graph \\ \hline
\multicolumn{1}{|g|}{$\nodeSet, \n$} & set of nodes, number of nodes \\ \hline
\multicolumn{1}{|g|}{$\edgeSet, \m$} & set of edges, number of edges \\ \hline 
\multicolumn{1}{|g|}{$sim(G_{1},G_{2})$}& similarity between graphs $G_{1}$ and $G_{2}$\\\hline
\multicolumn{1}{|g|}{$d(G_{1},G_{2})$}& distance between graphs $G_{1}$ and $G_{2}$\\\hline \hline
\multicolumn{1}{|g|}{$\matI$} & $\n \times \n$ identity matrix  \\ \hline
\multicolumn{1}{|g|}{$\matA$} & $\n \times \n$ adjacency matrix with elements $a_{ij}$ \\ \hline
\multicolumn{1}{|g|}{$\matD$} & $\n \times \n$ diagonal degree matrix, $d_{ii} = \sum_{j}a_{ij}$ \\ \hline
\multicolumn{1}{|g|}{$\matL$} & $=\matD-\matA$ laplacian matrix  \\ \hline 
\multicolumn{1}{|g|}{$\matB$} & $\n \times \n$ matrix of final scores with elements $s_{ij}$ \\ \hline
\multicolumn{1}{|g|}{$\matB'$} & $\n \times g$ reduced matrix of final scores \\ \hline \hline
\multicolumn{1}{|g|}{$\unitv_{i}$} & $\n \times 1$ unit vector with 1 in the $i^{th}$ element \\ \hline
\multicolumn{1}{|g|}{$\bhalf_{0k}$} & $\n \times 1$ vector of seed scores for group $k$ \\ \hline
\multicolumn{1}{|g|}{$\bhalf_i$} & $\n \times 1$ vector of final affinity scores to node $i$ \\ \hline \hline
\multicolumn{1}{|g|}{$g$} & number of groups (node partitions) \\ \hline
\multicolumn{1}{|g|}{$\epsilon$} & $=1/(1+\max_i{(d_{ii}))} $ positive constant ($<1$) \\
\multicolumn{1}{|g|}{} & encoding the influence between neighbors \\ \hline \hline
 \multicolumn{1}{|g|}{\myscshape{DC$_0$}, \myscshape{DC}}&\methodSlow, \method\\\hline
\multicolumn{1}{|g|}{\myscshape{VEO} }&Vertex/Edge Overlap\\\hline
\multicolumn{1}{|g|}{\myscshape{GED} }& Graph Edit Distance \cite{Bunke06}\\\hline
\multicolumn{1}{|g|}{\myscshape{SS} }& Signature Similarity \cite{PapadimitriouGM08}\\\hline
\multicolumn{1}{|g|}{\myscshape{$\lambda$-d Adj. / Lap }}&$\lambda$-distance on $\matA$ / $\matL$ / \\
\multicolumn{1}{|g|}{\myscshape{N.L.}}&normalized $\matL$ \\\hline
\hline\end{tabular}
}
\end{table} 

%
How can we find the similarity in connectivity between two graphs, or more formally how can we solve the following problem?
\vspace{-0.2cm}
\begin{problem}{\underline{\method}nectivity}
\begin{description*}
 \item [Given:] (a) two graphs, $G_1(\nodeSet, \edgeSet_{1})$ and $G_2(\nodeSet, \edgeSet_{2})$ with the same node set\footnote{If the graphs have different, but overlapping, node sets $\nodeSet_1$ and $\nodeSet_2$, we assume that $\nodeSet = \nodeSet_1 \cup \nodeSet_2$, and the extra nodes are treated as singletons.}, $\nodeSet$, and different edge sets $\edgeSet_{1}$ and $\edgeSet_{2}$, and (b) the node correspondence.
\item [Find:] a similarity score, $sim(G_{1},G_{2}) \in [0,1]$, between the input graphs. Similarity score of value 0 means totally different graphs, while 1 means identical graphs.
\end{description*}
\end{problem}
\noindent The obvious way to solve this problem is by measuring the overlap of their edges.
Why does this often not work in practice? Consider the following example: according to the overlap method, the pairs of barbell graphs shown in Fig.~\ref{fig:BPtests} of p.~\pageref{fig:BPtests}, $(B10, mB10)$ and $(B10, mmB10)$, have the same similarity score. But, clearly, from the aspect of information flow, a missing edge from a clique $(mB10)$ does not play as important role in the graph connectivity as the missing ``bridge'' in $mmB10$. So, could we instead measure the differences in the 1-step away neighborhoods, 2-step away neighborhoods etc.? If yes, with what weight? It turns out (Intuition \ref{obs:neighboringInfluence}) that our method does exactly this in a principled way.

\subsection{Fundamental Concept}
The first conceptual step of our proposed method is 
to compute the pairwise node affinities in the first graph, 
and compare them with the ones in the second graph.
For notational compactness,
we store them in a $n \times n$ 
similarity matrix\footnote{In reality, we don't measure all the affinities (see Section \ref{sec:fast} for an efficient approximation).} $\matB$. 
The $s_{ij}$ entry of the matrix  indicates the influence node $i$ has on node $j$. For example, in a who-knows-whom network, 
if node $i$ is, say, republican and if we assume homophily (i.e., neighbors are similar), 
how likely  is it  that node $j$ is also republican? 
Intuitively, node $i$ has more influence/affinity to node $j$ 
if there are many, short, heavily weighted paths from node $i$ to $j$. 

The second conceptual step is to measure the differences in the corresponding node affinity scores of the two graphs and report the result as their similarity score.

\subsection{How to measure node affinity?}
Pagerank, personalized Random Walks with Restarts (\rwr), lazy \rwr, and the ``electrical network analogy'' technique are only a few of the methods that compute node affinities. We could have used Personalized \rwr:
$
[\matI-(1-c)\matA\matD^{-1}]\bhalf_i = c\;\unitv_{i}, 
$
where $c$ is the probability of restarting the random walk from the initial node, $\unitv_i$ the starting (seed) indicator vector (all zeros except 1 at position $i$), and $\bhalf_i$ the unknown Personalized Pagerank column vector. Specifically, $s_{ij}$ is the affinity of node $j$ w.r.t. node $i$. 
For reasons that we explain next, we chose to use a more recent and principled method, the so-called
 Fast Belief Propagation (\fabp), 
 which is identical to Personalized \rwr under specific conditions (see Theorem \ref{thm:rwrfabp} in Appendix \ref{rwrfabp} of \cite{KoutraVF13appendix}).
%
We use a simplified form of it
(see Appendix \ref{sec:fabpToDeltacon} in \cite{KoutraVF13appendix})
given by:
\begin{equation}
\label{eq:fabp}
[\matI+\epsilon^{2}\matD-\epsilon\matA]\bhalf_i= \unitv_{i}
\end{equation}
where $\bhalf_i = [s_{i1}, ... s_{in}]^T$ is the column vector 
of final similarity/influence scores starting from the $i^{th}$ node, 
$\epsilon$ is a small constant capturing the influence between neighboring nodes, 
$\matI$ is the identity matrix,
$\matA$ is the adjacency matrix
and $\matD$ is the diagonal matrix with the degree of node $i$ as the $d_{ii}$ entry. 

\reminder{Danai, do you want to put eq3 in a box?}
An equivalent, more compact notation, is to use a matrix form,
and to stack all the $\bhalf_i$ vectors ($i=1, \ldots, n$)
into 
the $n \times n$ matrix $\matB$.
We can easily prove that
\begin{equation}
\label{eq:Smatrix}
\boxed{\matB = [s_{ij}] = [\matI+\epsilon^{2}\matD-\epsilon\matA]^{-1}}.
\end{equation}


\subsection{Why use Belief Propagation?}
The reasons  we choose \bp and its fast approximation
with Eq.~(\ref{eq:Smatrix}) are: (a) it is based on sound theoretical background (maximum likelihood estimation on marginals), 
(b) it is fast (linear on the number of edges), and (c) it 
agrees with intuition, taking into account
not only direct neighbors, but also 2-, 3- and $k$-step-away
neighbors, with decreasing weight.
We elaborate on the last reason, next:
\begin{intuition}{[\textbf{Attenuating Neighboring Influence}]} \\
By temporarily ignoring the term $\epsilon^2 \matD$ in  (\ref{eq:Smatrix}), 
we can expand the matrix inversion
and approximate the $n \times n$ matrix of pairwise affinities, $\matB$, as
\begin{equation*}
\matB \approx [\matI - \epsilon \matA]^{-1} 
     \approx \matI + \epsilon \matA + \epsilon^2 \matA^2 + \ldots.
\label{obs:neighboringInfluence}
\end{equation*}
\end{intuition}
As we said, our method captures the differences in the 1-step, 2-step, 3-step etc. neighborhoods in a weighted way; differences in long paths have smaller effect on the computation of the similarity measure than differences in short paths. Recall that $\epsilon < 1$, and 
that $\matA^k$ has information about the $k$-step paths. Notice that this is just the intuition behind our method; we do not use this simplified formula to find matrix $\matB$.

\subsection{Which properties should a similarity measure satisfy?}

Let $G_1(\nodeSet, \edgeSet_{1})$ and $G_2(\nodeSet, \edgeSet_{2})$ be two graphs, and $sim(G_1, G_2) \in [0,1]$ denote their similarity score. Then, we want the similarity measure to obey the following axioms:
\begin{itemize*}
\item[A1.] \emph{Identity property}: $sim(G_1, G_1) = 1$
\item[A2.] \emph{Symmetric property}: $sim(G_1, G_2) = sim(G_2, G_1)$
\item[A3.] \emph{Zero property}: $sim(G_1, G_2)  \to 0$ for $\n \to \infty$, where $G_1$ is the clique graph ($K_{\n}$), and $G_2$ is the empty graph (i.e., the edge sets are complementary).
\end{itemize*}

\noindent Moreover, the measure must be: 
\paragraph*{\textbf{(a) intuitive.}} It should satisfy the following desired properties:
\begin{enumerate*}
 \item[P1.] [\emph{Edge Importance}] Changes that create disconnected components should be penalized more than changes that maintain the connectivity properties of the graphs.
 \item[P2.] [\emph{Weight Awareness}] In weighted graphs, the bigger the weight of the removed edge is, the greater the impact on the similarity measure should be.
 \item[P3.] [\emph{Edge-``Submodularity''}] A specific change is more important in a graph with few edges than in a much denser, but equally sized graph.
 \item[P4.] [\emph{Focus Awareness}] Random changes in graphs are less important than targeted changes of the same extent.
 \end{enumerate*}
\paragraph*{\textbf{(b) scalable.}} The huge size of the generated graphs, as well as their abundance require a  similarity measure that is computed fast and handles graphs with billions of nodes. 

\section{Proposed Method: Details}
\label{sec:proposedDetails}
Now that we have described the high level ideas behind our method, we move on to the details.

\subsection{Algorithm Description}
\label{sec:slow}

 
Let the graphs we compare be $G_1(\nodeSet, \edgeSet_1)$ and $G_2(\nodeSet, \edgeSet_2)$.  If the graphs have different node sets, say $\nodeSet_{1}$ and $\nodeSet_{2}$, we assume that $\nodeSet = \nodeSet_{1} \cup \nodeSet_{2}$, where some nodes are disconnected.

As mentioned before, the main idea behind our proposed similarity algorithm is to compare the node affinities in the given graphs. The steps of our similarity method are:
\paragraph*{\bf Step 1.} By eq.~(\ref{eq:Smatrix}), we compute for each graph the $\n \times \n$ matrix of pairwise node affinity scores ($\matB_1$ and $\matB_2$ for graphs $G_1$ and $G_2$ respectively). 
\paragraph*{\bf Step 2.} Among the various distance and similarity measures (e.g., Euclidean distance (ED),  cosine similarity, correlation) found in the literature, we use the root euclidean distance (\matusita, a.k.a. Matusita distance)
\begin{equation}
\label{eq:matusita}
d =  \matusita(\matB_1, \matB_2) = \sqrt{\sum_{i=1}^{\n}\sum_{j=1}^{\n}{ (\sqrt{s_{1,ij}} - \sqrt{s_{2,ij}} })^2 }.
\end{equation}
We use the \matusita distance for the following reasons:
\vspace{-0.2cm}
\begin{enumerate*}
\item it is very similar to the Euclidean distance (ED), the only difference being the square root 
of the pairwise similarities ($s_{ij}$),
\item it usually gives better results, because it ``boosts'' the node affinities\footnote{The node affinities are in $[0,1]$, so the square root makes them bigger.} and, therefore, detects even small changes in the graphs (other distance measures, including ED, suffer from high similarity scores no matter how much the graphs differ), and  
\item satisfies the desired properties $P1$-$P4$. As discussed in the Appendix \ref{sec:properties} of \cite{KoutraVF13appendix}, at least $P1$ is not satisfied by the ED.
\end{enumerate*}
%
\paragraph*{\bf Step 3.} For interpretability,
we convert the distance ($d$) to similarity measure ($sim$) via the formula $sim = \frac{1}{1+d}$. 
The result is bounded to the interval [0,1], as opposed to being unbounded [0,$\infty$).
Notice that the distance-to-similarity transformation
does {\em not} change the ranking of results in a nearest-neighbor query.

The straightforward algorithm, \methodSlow (Algorithm \ref{alg:slow}), is to compute all the $\n^{2}$ affinity scores of matrix $\matB$ by simply using equation (\ref{eq:Smatrix}). We can do the inversion using the Power Method or any other efficient method.
\begin{algorithm}
\caption{\methodSlow}
\label{alg:slow}
\begin{algorithmic}
\STATE INPUT: edge files of $G_1(\nodeSet, \edgeSet_1)$ and $G_2(\nodeSet, \edgeSet_2)$ 
\STATE // $\nodeSet = \nodeSet_{1} \cup \nodeSet_{2}$, if $\nodeSet_{1}$ and $\nodeSet_{2}$ are the graphs' node sets
\vspace{0.1cm}
\STATE $\mathbf{\matB_{1}}=[\matI+\epsilon^{2}\matD_1-\epsilon\matA_1]^{-1}$ \hfill // $s_{1,ij}$: affinity/influence of
\STATE $\mathbf{\matB_{2}}=[\matI+\epsilon^{2}\matD_2-\epsilon\matA_2]^{-1}$ \hfill //node $i$ to node $j$ in  $G_1$
\STATE $d(G_{1}, G_{2})= $\matusita($\matB_1, \matB_2$) 
\RETURN $sim(G_1, G_2) = \frac{1}{1+ d }$
\end{algorithmic}
\end{algorithm}	
\vspace{-0.2cm}
\subsection{Scalability Analysis}
\label{sec:fast}
\methodSlow 
satisfies all the properties in Section \ref{sec:proposed}, 
but it is quadratic ($\n^2$ affinity scores $s_{ij}$ - using power method for the inversion of sparse matrix)
and thus not scalable.
We present a faster, linear algorithm, \method (Algorithm \ref{alg:fast}), which approximates \methodSlow and differs in the first step. We still want each node to become a seed exactly once in order to find the affinities of the rest of the nodes to it; but, here, we have multiple seeds at once, instead of having one seed at a time.
The idea is to randomly divide our node-set into $g$ groups, and compute
the affinity score of each node $i$ to group $k$, thus requiring
only $\n \times g$ scores, which are stored in the $\n \times g$ matrix $\matB'$ ($g \ll n$).
Intuitively, instead of using the $\n \times \n$ affinity matrix $\matB$,
we add up the scores of the columns that correspond to the nodes of a group, 
and obtain the $n \times g$ matrix $\matB'$ ($g \ll n$).
The score $s'_{ik}$ is the affinity of node $i$ to the $k^{th}$ {\em group} of nodes ($k=1,\ldots,g$).
\begin{lemma}
\label{lem:reduced}
The time complexity of computing the reduced affinity matrix, $\matB'$, is linear on the number of edges.
\end{lemma}
\begin{proof}
We can compute the $\n \times g$ ``skinny'' matrix $\matB'$ quickly,
by  solving $[\matI+\epsilon^{2}\matD-\epsilon\matA]\matB'= [\bhalf_{01} \; \ldots \bhalf_{0g}],$ 
where $\bhalf_{0k}=\sum_{i \in group_{k}}\unitv_{i}$ is the membership $n \times 1$ vector for group $k$ (all 0's, except 1's  for members of the group). \hfill $\square$
\end{proof}
Thus, we compute $g$ final scores per node, which denote its affinity to every \emph{group} of seeds, instead of every seed node that we had in eq.~(\ref{eq:Smatrix}). 
With careful implementation, \method
is linear on the number of number of edges and groups $g$. As we show in section \ref{subsect:scalability},
it takes $\sim160$sec, on commodity hardware, for a 1.6-million-node graph.

Once we have the reduced affinity matrices $\matB'_1$ and $\matB'_2$ of the two graphs,
we use the \matusita, to find the similarity between the $\n \times g$ matrices 
of final scores, where $g \ll n$.


\begin{lemma}
The time complexity of \method, when applied in parallel to the input graphs, is linear on the number of edges in the graphs, i.e. O$(g \cdot max\{\m_{1}, \m_{2} \})$.
\label{lem:complexity}
\end{lemma}
\begin{proof}
Based on lemma \ref{lem:reduced}. See Appendix \ref{sec:details} in \cite{KoutraVF13appendix}. \hfill $\square$
\end{proof}

\begin{thm}
\method's similarity score between any two graphs $G_{1}$, $G_{2}$ upper bounds the actual \methodSlow's similarity score, i.e. $sim_{DC-0}(G_{1}, G_{2}) \le sim_{DC}(G_{1}, G_{2})$.
\label{thm:upperBound}
\end{thm}
\begin{proof}
Intuitively, grouping nodes blurs the influence information and makes the nodes seem more similar than originally. For more details, see Appendix \ref{sec:details} of \cite{KoutraVF13appendix}. \hfill $\square$
\end{proof}


In the following section we show that \method (which includes \methodSlow as a special case for $g=n$) satisfies the axioms and properties, while in the Appendix (\ref{sec:axioms} and \ref{sec:properties} in \cite{KoutraVF13appendix}) we provide the proofs.

\begin{algorithm}
\caption{\method} 
\label{alg:fast}
\begin{algorithmic}
\STATE INPUT: edge files of $G_1(\nodeSet, \edgeSet_1)$ and $G_2(\nodeSet, \edgeSet_2)$ and
\STATE $\;\;\;\;\;\;\;\;\;\;\;$      $g$ (groups: \# of node partitions)
\vspace{0.2cm}
\STATE $\{\nodeSet_{j}\}_{j=1}^{g}$ = random\_partition($\nodeSet, g$) \hfill //{$g$ groups}
\STATE // estimate affinity vector of nodes $i=1,\ldots,n$ to group $k$ 
\FOR {$k = 1 \to g$}
\STATE $\phihalf_{k} = \sum_{i \in\nodeSet_{k}}\unitv_{i}$
\STATE solve $[\matI+\epsilon^{2}\matD_1-\epsilon\matA_1]\bhalf_{1k}'={\phihalf_{k}}$
\STATE solve $[\matI+\epsilon^{2}\matD_2-\epsilon\matA_2]\bhalf_{2k}'={\phihalf_{k}}$
\ENDFOR \\
\STATE  $\matB'_{1} = [ \bhalf_{11}' \;  \bhalf_{12}' \; \ldots \; \bhalf_{1g}' ]; \;\; \matB'_{2} = [ \bhalf_{21}' \; \bhalf_{22}' \;  \ldots  \; \bhalf_{2g}']$
\STATE // compare affinity matrices $\matB_1'$ and $\matB_2'$
\STATE $d(G_{1}, G_{2})= $\matusita$(\matB'_1, \matB'_2)$
\RETURN $sim(G_1, G_2) = \frac{1}{1+ d }$
\end{algorithmic}
\end{algorithm}

\vspace{-0.4cm}
\section{Experiments}
\label{sec:experiments}
We conduct several experiments on synthetic and real data (undirected, unweighted graphs, unless stated otherwise - see Table  \ref{tab:datasets}) to answer the following questions:
\begin{enumerate*}
\item[{\bf Q1.}] Does \method agree with our intuition and satisfy the axioms/properties? Where do other methods fail?
\item[{\bf Q2.}] Is \method scalable?
\end{enumerate*}


\begin{table}
\centering
\caption{Real and Synthetic Datasets}
\label{tab:datasets}
\small{ \begin{tabular}{|l||r|r|c| } \hline
\multicolumn{1}{|g||}{\textbf{Name}} & \multicolumn{1}{g}{\textbf{Nodes}} & \multicolumn{1}{|g}{\textbf{Edges}} & \multicolumn{1}{|g|}{\textbf{Description}}\\ \hline \hline
\multicolumn{1}{|g||}{\textbf{Brain Graphs}} & 70 & 800-1,208 & \scriptsize{connectome}\\ \hline
\multicolumn{1}{|g||}{\textbf{Enron Email} \cite{KlimtY04}} & 36,692 & 367,662 & \scriptsize{who-emails-whom}\\ \hline
\multicolumn{1}{|g||}{\textbf{Epinions} \cite{GuhaKRT04}} & 131,828 & 841,372 & \scriptsize{who-trusts-whom}\\ \hline
\multicolumn{1}{|g||}{\textbf{Email EU} \cite{LeskovecKF07}} & 265,214 & 420,045 & \scriptsize{who-sent-to-whom}\\ \hline
\multicolumn{1}{|g||}{\textbf{Web Google} \cite{snap_web}} & 875,714 & 5,105,039 & \scriptsize{site-to-site}\\ \hline 
\multicolumn{1}{|g||}{\textbf{AS skitter} \cite{LeskovecKF07}} &1,696,415 & 11,095,298 & \scriptsize{p2p links}\\ \hline \hline
\multicolumn{1}{|g||}{\textbf{Kronecker 1}} & 6,561 & 65,536 & \scriptsize{synthetic}\\ \hline
\multicolumn{1}{|g||}{\textbf{Kronecker 2}} & 19,683 & 262,144 & \scriptsize{synthetic}\\ \hline
\multicolumn{1}{|g||}{\textbf{Kronecker 3}} & 59,049 & 1,048,576 & \scriptsize{synthetic}\\ \hline
\multicolumn{1}{|g||}{\textbf{Kronecker 4}} & 177,147 & 4,194,304 & \scriptsize{synthetic}\\ \hline
\multicolumn{1}{|g||}{\textbf{Kronecker 5}} & 531,441 & 16,777,216 & \scriptsize{synthetic}\\ \hline
\multicolumn{1}{|g||}{\textbf{Kronecker 6}} & 1,594,323 & 67,108,864 & \scriptsize{synthetic}\\ \hline
\end{tabular}
}
\end{table}

\noindent The implementation is in Matlab and we ran the experiments on AMD Opteron Processor 854 @3GHz, RAM 32GB. \reminder{(RAM speed -- what about amstel??)} 

\subsection{Intuitiveness of \methodT.}
\label{sec:intuitive}

\begin{figure*}[tbh!]
	  \begin{minipage}{0.65\textwidth}
	    \centering
	    \includegraphics[width=\textwidth]{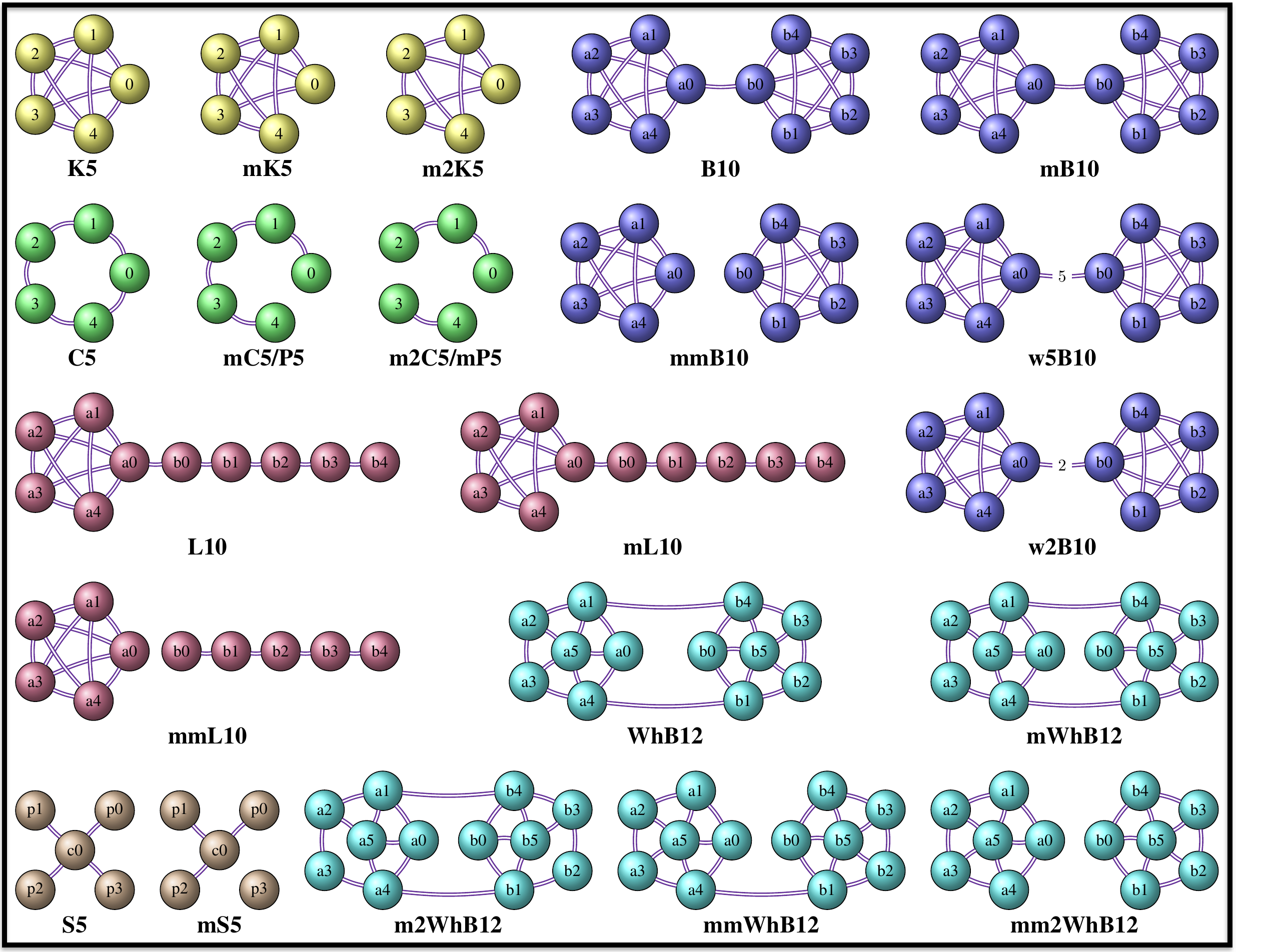}
	    \captionof{figure}{Small synthetic graphs  -- \emph{K: clique, C: cycle, P: path, S: star, B: barbell, L: lollipop, WhB: wheel-barbell}}
	    	    \label{fig:BPtests}
	  \end{minipage} \hfill
	  \begin{minipage}{0.32\textwidth}
	    \centering
	        
	\vspace{5.5cm}
	{\scriptsize\begin{tabular}{|c|c|}
	\hline
	\rowcolor{Gray}\textbf{Symbol}& \textbf{Meaning}\\\hline
	\multicolumn{1}{|g|}{$\mathbf{K_{n}}$}&clique of size $n$\\\hhline{--}
	\multicolumn{1}{|g|}{$\mathbf{P_{n}}$}&path of size $n$\\\hline
	\multicolumn{1}{|g|}{$\mathbf{C_{n}}$}&cycle of size $n$\\\hhline{--}
	\multicolumn{1}{|g|}{$\mathbf{S_{n}}$}&star of size $n$\\\hline
	\multicolumn{1}{|g|}{$\mathbf{L_{n}}$}&lollipop of size $n$\\\hhline{--}
	\multicolumn{1}{|g|}{$\mathbf{B_{n}}$}&barbell of size $n$\\\hline
	\multicolumn{1}{|g|}{$\mathbf{WhB_{n}}$}&wheel barbell of size $n$\\\hline
	\multicolumn{1}{|g|}{$\mathbf{m_{X}}$} & missing X edges\\\hhline{--}
	\multicolumn{1}{|g|}{$\mathbf{mm_{X}}$} & missing X ``bridge'' edges\\\hhline{--}
	\multicolumn{1}{|g|}{$\textbf{w}$} & weight of ``bridge'' edge\\\hhline{--}
	\end{tabular}
\label{tab:nameConventions}	

	}
		\vspace{0.1cm}

\captionof{table}{Name Conventions for small synthetic graphs. Missing number after the prefix implied $X=1$.}
	     \end{minipage}
	\end{figure*}


\begin{figure*}[tbh!]
\begin{minipage}{0.5\textwidth}
\centering
\captionof{table}{``Edge Importance'' (P1). Highlighted entries violate P1.
 }
\label{tab:edgeImportance}
\resizebox{\columnwidth}{!} {
\begin{tabular}{|c|c|c|||c|c|c|c||c|c|c|c|}
\hline \multicolumn{3}{|c|||}{} & & & & & \myscshape{GED} &   \myscshape{$\lambda$-d}& \myscshape{$\lambda$-d}&  \myscshape{$\lambda$-d} \\
\multicolumn{3}{|c|||}{ \multirow{-2}{*}{\myscshape{Graphs}}}&  \multirow{-2}{*}{ \myscshape{DC$_0$}} & \multirow{-2}{*}{ \myscshape{DC}} &\multirow{-2}{*}{ \myscshape{VEO}}&\multirow{-2}{*}{ \myscshape{SS}}& \myscshape{(XOR)} & \myscshape{Adj.}&  \myscshape{Lap.}&  \myscshape{N.L.} \\ \hline 
\multicolumn{1}{|c}{ \myscshape{A}}&\multicolumn{1}{|c}{ \myscshape{B}}&\multicolumn{1}{|c|||}{ \myscshape{C}}&  \multicolumn{4}{c||}{$\Delta{s} = sim(A,B) - sim(A,C)$}&   \multicolumn{4}{c|}{$\Delta{d} = d(A,C) - d(A,B)$} \\ \hline \hline  \hline
B10&mB10&mmB10&{\bf 0.07}&{\bf 0.04}&\wrong{0}&\wrong{$-10^{-5}$}&\wrong{0}&0.21&\wrong{-0.27}&2.14\\\hline
L10&mL10&mmL10&{\bf 0.04}&{\bf 0.02}&\wrong{0}&$10^{-5}$&\wrong{0}&\wrong{-0.30}&\wrong{-0.43}&\wrong{-8.23}\\\hline
WhB10&mWhB10&mmWhB10&{\bf 0.03}&{\bf 0.01}&\wrong{0}&\wrong{$-10^{-5}$}&\wrong{0}&0.22&0.18&\wrong{-0.41}\\\hline
WhB10&m2WhB10&mm2WhB10&{\bf 0.07}&{\bf 0.04}&\wrong{0}&\wrong{$-10^{-5}$}&\wrong{0}&0.59&\wrong{0.41}&0.87\\\hline
\end{tabular}
}
\centering
\vspace{0.5cm}
\setcounter{table}{5}
\captionof{table}{``Edge-Submodularity'' (P3). Highlighted entries violate P3.
}
\label{tab:submodularity}
\resizebox{\columnwidth}{!} {
\begin{tabular}{|c|c||c|c|||c|c|c|c||c|c|c|c|}
\hline \multicolumn{4}{|c|||}{} & & & & & \myscshape{GED} &   \myscshape{$\lambda$-d}& \myscshape{$\lambda$-d}&  \myscshape{$\lambda$-d} \\
\multicolumn{4}{|c|||}{ \multirow{-2}{*}{\myscshape{Graphs}}}&  \multirow{-2}{*}{ \myscshape{DC$_0$}}& \multirow{-2}{*}{ \myscshape{DC}}&\multirow{-2}{*}{ \myscshape{VEO}}&\multirow{-2}{*}{ \myscshape{SS}}&\myscshape{(XOR)}&   \myscshape{Adj.}&  \myscshape{Lap.}&  \myscshape{N.L.}\\ \hline 
\multicolumn{1}{|c}{ \myscshape{A}}&\multicolumn{1}{|c||}{ \myscshape{B}}&\multicolumn{1}{c|}{ \myscshape{C}}&\multicolumn{1}{|c|||}{ \myscshape{D}}&  \multicolumn{4}{c||}{$\Delta{s} = sim(A,B) - sim(C,D)$}&   \multicolumn{4}{c|}{$\Delta{d} = d(C,D) - d(A,B)$} \\ \hline \hline \hline
K5&mK5&C5&mC5&{\bf 0.03}&{\bf 0.03}&0.02&$10^{-5}$&\wrong{0}&\wrong{-0.24}&\wrong{-0.59}&\wrong{-7.77}\\\hline
C5&mC5&P5&mP5&{\bf 0.03}&{\bf 0.01}&0.01&\wrong{$-10^{-5}$}&\wrong{0}&\wrong{-0.55}&\wrong{-0.39}&\wrong{-0.20}\\\hline
P5&mP5&S5&mS5&{\bf 0.003}&{\bf 0.001}&\wrong{0}&\wrong{$-10^{-5}$}&\wrong{0}&\wrong{-0.07}&0.39&3.64\\\hline
K\tiny{100}&mK\tiny{100}&C\tiny{100}&mC\tiny{100}&{\bf 0.03}&{\bf 0.02}&0.002&$10^{-5}$&\wrong{0}&\wrong{-1.16}&\wrong{-1.69}&\wrong{-311}\\\hline
C\tiny{100}&mC\tiny{100}&P\tiny{100}&mP\tiny{100}&{ $\mathbf{10^{-4}}$}&{\bf 0.01}&$10^{-5}$&\wrong{$-10^{-5}$}&\wrong{0}&\wrong{-0.08}&\wrong{-0.06}&\wrong{-0.08}\\\hline
P\tiny{100}&mP\tiny{100}&S\tiny{100}&mS\tiny{100}&{\bf 0.05}&{\bf 0.03}&\wrong{0}&\wrong{0}&\wrong{0}&\wrong{-0.08}&1.16&196\\\hline
K\tiny{100}&m10K\tiny{100}&C\tiny{100}&m10C\tiny{100}&{\bf 0.10}&{\bf 0.08}&0.02&$10^{-5}$&\wrong{0}&\wrong{-3.48}&\wrong{-4.52}&\wrong{-1089}\\\hline
C\tiny{100}&m10C\tiny{100}&P\tiny{100}&m10P\tiny{100}&{\bf 0.001}&{\bf 0.001}&$10^{-5}$&\wrong{0}&\wrong{0}&\wrong{-0.03}&0.01&0.31\\\hline
P\tiny{100}&m10P\tiny{100}&S\tiny{100}&m10S\tiny{100}&{\bf 0.13}&{\bf 0.07}&\wrong{0}&\wrong{$-10^{-5}$}&\wrong{0}&\wrong{-0.18}&8.22&1873\\\hline
\end{tabular}
}
\end{minipage} \hfill
\begin{minipage}{0.5\textwidth}
\setcounter{table}{4}
%
\centering
\captionof{table}{``Weight Awareness'' (P2). Highlighted entries violate P2.
}
\label{tab:weightAwareness}
\resizebox{\columnwidth}{!} {
\begin{tabular}{|c|c||c|c|||c|c|c|c||c|c|c|c|}
\hline \multicolumn{4}{|c|||}{} & & & & & \myscshape{GED} &   \myscshape{$\lambda$-d}& \myscshape{$\lambda$-d}&  \myscshape{$\lambda$-d} \\
\multicolumn{4}{|c|||}{ \multirow{-2}{*}{\myscshape{Graphs}}}&  \multirow{-2}{*}{ \myscshape{DC$_0$}}& \multirow{-2}{*}{ \myscshape{DC}}&\multirow{-2}{*}{ \myscshape{VEO}}&\multirow{-2}{*}{ \myscshape{SS}}&\myscshape{(XOR)}&   \myscshape{Adj.}&  \myscshape{Lap.}&  \myscshape{N.L.} \\ \hline 
\multicolumn{1}{|c}{ \myscshape{A}}&\multicolumn{1}{|c||}{ \myscshape{B}}&\multicolumn{1}{|c}{ \myscshape{C}}&\multicolumn{1}{|c|||}{ \myscshape{D}}&  \multicolumn{4}{c||}{$\Delta{s} = sim(A,B) - sim(C,D)$}&   \multicolumn{4}{c|}{$\Delta{d} = d(C,D) - d(A,B)$} \\ \hline \hline \hline
B10&mB10&B10&w5B10&{\bf 0.09}&{\bf 0.08}&\wrong{-0.02}&$10^{-5}$&\wrong{-1}&3.67&5.61&84.44\\\hline
mmB10&B10&mmB10&w5B10&{\bf 0.10}&{\bf 0.10}&\wrong{0}&$10^{-4}$&\wrong{0}&4.57&7.60&95.61\\\hline
B10&mB10&w5B10&w2B10&{\bf 0.06}&{\bf 0.06}&\wrong{-0.02}&$10^{-5}$&\wrong{-1}&2.55&3.77&66.71\\\hline
w5B10&w2B10&w5B10&mmB10&{\bf 0.10}&{\bf 0.07}&0.02&$10^{-5}$&1&2.23&3.55&31.04\\\hline
w5B10&w2B10&w5B10&B10&{\bf 0.03}&{\bf 0.02}&\wrong{0}&$10^{-5}$&\wrong{0}&1.12&1.84&17.73\\\hline
\end{tabular}
}

\vspace{0.1cm}

	\captionof{figure}{``Focus-Awareness'' (P4).}
	\includegraphics[width=0.82\columnwidth]{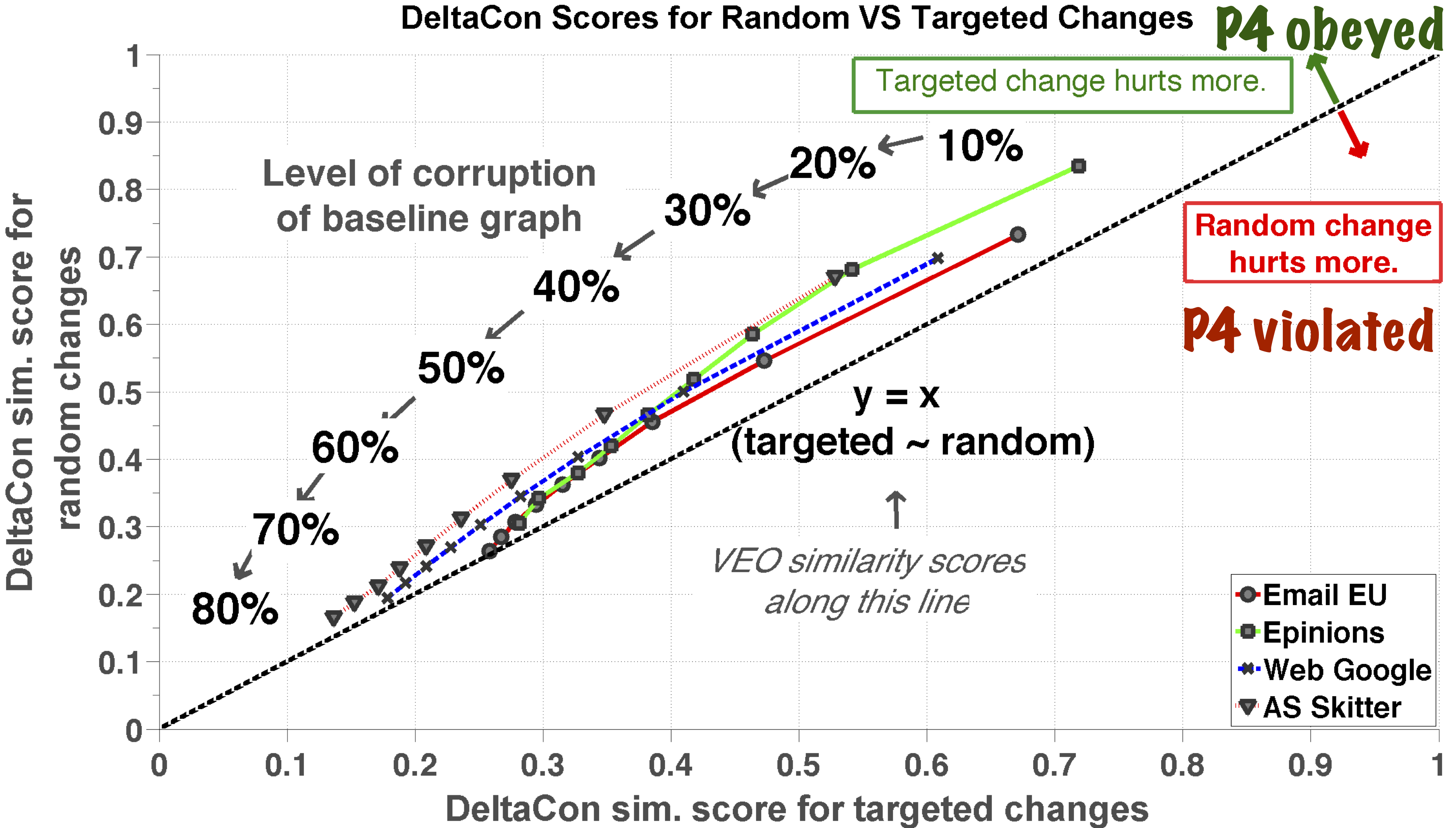}
	 \label{fig:real_sim_randVStarg}
	    	    \label{fig:real_sim_randVStarg}
\end{minipage}
\vspace{0.2cm}
\caption*{Tables 4-6/Figure 3: \methodSlow and \method (in bold) obey all the required properties (P1-P4). 
\emph{Tables 4-6}: Each row of the tables corresponds to a comparison between the similarities (or distances) of two pairs of graphs; pairs (A,B) and (A,C) for (P1); and pairs (A,B) and (C,D) for (P2) and (P3): Non-positive values of {\scriptsize $\Delta{s} = sim(A,B) - sim(C,D)$} and {\scriptsize $\Delta{d} = d(C,D) - d(A,B)$} - depending on whether the corresponding method computes similarity or distance - are highlighted and mean violation of the property of interest. 
\emph{Figure 3}: Targeted changes hurt more than random. Plot of \method similarity scores for random changes (y axis) vs. \method similarity scores for targeted changes (x axis) for 4 real-world networks. For each graph we create 8 ``corrupted'' versions with 10\% to 80\% fewer edges than the initial graphs. Notice that all the points are above the diagonal.}
\end{figure*}

\setcounter{table}{6}

To answer Q1, for the first 3 properties (P1-P3), we conduct experiments on small graphs of 5 to 100 nodes and classic topologies (cliques, stars, circles, paths, barbell and wheel-barbell graphs,  and ``lollipops'' shown in Fig.~\ref{fig:BPtests}), since people can argue about their similarities. For the name conventions see Table 3. For our method we used 5 groups ($g$), but the results are similar for other choices of the parameter. In addition to the synthetic graphs, for the last property (P4), we use real networks with up to 11 million edges (Table \ref{tab:datasets}).

We compare our method, \method, to the 6 best state-of-the-art similarity measures that apply to our setting:
\begin{enumerate*}
\item Vertex/Edge Overlap (VEO) \cite{PapadimitriouGM08}: For two graphs $G_{1}(\nodeSet_{1}, \edgeSet_{1})$ and $G_{2}(\nodeSet_{2}, \edgeSet_{2})$:
\begin{equation*}
sim_{VEO}(G_{1}, G_{2})=2\frac{|\edgeSet_{1} \cap \edgeSet_{2}|+|\nodeSet_{1} \cap \nodeSet_{2}|}{|\edgeSet_{1}| + |\edgeSet_{2}| + |\nodeSet_{1}| + |\nodeSet_{2}|}.
\label{eq:veo}
\end{equation*}
\item Graph Edit Distance (GED) \cite{Bunke06}: GED has quadratic complexity in general, so they \cite{Bunke06} consider the case where only insertions and deletions are allowed.
\begin{eqnarray*}
sim_{GED}(G_{1}, G_{2}) &=& |\nodeSet_{1}| + |\nodeSet_{2}| - 2 |\nodeSet_{1} \cap \nodeSet_{2}|\\
 &+& |\edgeSet_{1}| + |\edgeSet_{2}| - 2 |\edgeSet_{1} \cap \edgeSet_{2}|.
\label{eq:ged}
\end{eqnarray*}
For $\nodeSet_{1} = \nodeSet_{2}$ and unweighted graphs,~$sim_{GED}$ is equivalent to $\text{hamming distance}(\mathbf{A_1}, \mathbf{A_2})=sum(\mathbf{A_1}\,XOR\,\mathbf{A_2})$. 
\item Signature Similarity (SS) \cite{PapadimitriouGM08}: This is the best performing similarity measure studied in \cite{PapadimitriouGM08}. It is based on the SimHash algorithm (random projection based method).
\item The last 3 methods are variations of the well-studied spectral method ``$\lambda$-distance'' (\cite{Bunke06}, \cite{Peabody03}, \cite{WilsonZ08}). Let \{$\lambda_{1i}\}_{i=1}^{|\nodeSet_{1}|}$ and \{$\lambda_{2i}\}_{i=1}^{|\nodeSet_{2}|}$ be the eigenvalues of the matrices that represent $G_{1}$ and $G_{2}$. Then, $\lambda$-distance is given by
\vspace{-0.4cm}
\begin{equation*}
d_{\lambda}(G_{1}, G_{2}) = \sqrt{\sum_{i=1}^{k}{(\lambda_{1i}-\lambda_{2i})}^{2}},
\label{eq:ldist}
\end{equation*}
where $k$ is $max(|\nodeSet_{1}|, |\nodeSet_{2}|)$ (padding is required for the smallest vector of eigenvalues).
The variations of the method are based on three different matrix representations of the graphs: adjacency ($\lambda$-d Adj.), laplacian ($\lambda$-d Lap.) and normalized laplacian matrix ($\lambda$-d N.L.).
\end{enumerate*}

The results for the first 3 properties are presented in the form of tables \ref{tab:edgeImportance}-\ref{tab:submodularity}. For property P1 we compare the graphs (A,B) and (A,C) and report the difference between the pairwise similarities/distances of our proposed methods and the 6 state-of-the-art methods. We have arranged the pairs of graphs in such way that (A,B) are more similar than (A,C). Therefore, table entries that are non-positive mean that the corresponding method does not satisfy the property. Similarly, for properties P2 and P3, we compare the graphs (A,B) and (C,D) and report the difference in their pairwise similarity/distance scores.


\paragraph{P1. Edge Importance}:
\emph{``Edges whose removal creates disconnected components are more important than other edges whose absence does not affect the graph connectivity. The more important an edge is, the more it should affect the similarity or distance measure.''}

For this experiment we use the barbell, ``wheel barbell'' and ``lollipop'' graphs, since it is easy to argue about the importance of the individual edges. The idea is that edges \emph{in} a highly connected component (e.g. clique, wheel) are not very important from the information flow viewpoint, while edges that \emph{connect} (almost uniquely) dense components play a significant role in the connectivity of the graph and the information flow. The importance of the ``bridge'' edge depends on the size of the components that it connects; the bigger the components the more important is the role of the edge. 

\begin{observation}
Only \method succeeds in distinguishing the importance of the edges $(P1)$ w.r.t. connectivity, while all the other methods fail at least once (Table \ref{tab:edgeImportance}).
\end{observation}


\paragraph{P2. Weight Awareness}:
\emph{``The absence of an edge of big weight is more important than the absence of a smaller weighted edge; this should be reflected in the similarity measure.''}

The weight of an edge defines the strength of the connection between two nodes, and, in this sense, can be viewed as a feature that relates to the importance of the edge in the graph. For this property, we study the weighted versions of the barbell graph, where we assume that all the edges except the ``bridge'' have unit weight. 

\reminder{maybe add some more results?}

\begin{observation}
All the methods are weight-aware $(P2)$, except VEO and GED which compute just the overlap in edges and vertices between the graphs (Table \ref{tab:weightAwareness}).
\end{observation}

\paragraph{P3. ``Edge-Submodularity''}:
%
\emph{``Let $A(\nodeSet,\edgeSet_{1})$ and $B(\nodeSet,\edgeSet_{2})$ be two graphs with the same node set, and $|\edgeSet_{1}|>|\edgeSet_{2}|$ edges. Also, assume that $m_{x}A(\nodeSet, \edgeSet_{1})$ and $m_{x}B(\nodeSet, \edgeSet_{2})$ are the respective derived graphs after removing $x$ edges. We expect that $sim(A, m_{x}A) > sim(B, m_{x}B)$, since the fewer the edges in a constant-sized graph, the more ``important'' they are.''}

The results for different graph topologies and 1 or 10 removed edges (prefixes 'm' and 'm10' respectively) are given compactly in Table \ref{tab:submodularity}. Recall that non-positive values denote violation of the ``edge-submodularity'' property. 

\begin{observation}
Only \method 
complies to the ``edge-submodularity'' property $(P3)$ in all cases examined.
\end{observation}

\paragraph{P4. Focus Awareness}:
%
At this point, all the competing methods have failed in satisfying at least one of the desired properties. 
To test whether \method is able to distinguish the extent of a change in a graph, we analyze real datasets with up to 11 million edges (Table \ref{tab:datasets}) for two different types of changes. For each graph we create corrupted instances by removing: (i) edges from the original graph randomly, and (ii) the same number of edges in a targeted way (we randomly choose nodes and remove all their edges, until we have removed the appropriate fraction of edges). 

In Fig.~\ref{fig:real_sim_randVStarg}, for each of the 4 real networks -Email EU, Enron, Google web and AS Skitter-, we give the pair (sim\_\method random, sim\_\method targeted) for each of the different levels of corruption (10\%, 20\%, \ldots, 80\%). That is, for each corruption level and network, there is a point with coordinates the similarity score between the original graph and the corrupted graph when the edge removal is random, and the score when the edge removal is targeted. 
The line $y=x$ corresponds to equal similarity scores for both ways of removing edges. 
%
\begin{observation}

\begin{itemize*}
\item \emph{``Targeted changes hurt more.''} \method is focus-aware $(P4)$. Removal of edges in a targeted way leads to smaller similarity of the derived graph to the original one than removal of the same number of edges in a random way.
\item \emph{``More changes: random $\approx$ targeted.''} 
As the corruption level increases, the similarity score for random changes tends to the similarity score for targeted changes (in Fig.~\ref{fig:real_sim_randVStarg}, all lines converge to the $y=x$ line for greater level of corruption). 
\end{itemize*} 
\end{observation}

%
\noindent This is expected as the random and targeted edge removal tend to be equivalent when a significant fraction of edges is deleted.

\mypar{General Remarks}
All in all, the baseline methods have several non-desirable properties. The spectral methods, as well as SS fail to comply to the ``edge importance'' (P1) and ``edge submodularity'' (P3) properties. Moreover, $\lambda$-distance has high computational cost when the whole graph spectrum is computed, cannot
distinguish the differences between co-spectral graphs, and sometimes small changes lead to big differences in the graph spectra. As far as VEO and GED are concerned, they are oblivious on significant structural properties of the graphs; thus, despite their straightforwardness and fast computation, they fail to discern various changes in the graphs.
On the other hand, \method gives tangible similarity scores and conforms to all the desired properties.

\subsection{Scalability of \methodT.}
\label{subsect:scalability}
In Section~\ref{sec:proposed} we demonstrated that \method is linear on the number of edges, and here we show that this also holds in practice.
%
We ran \method on Kronecker graphs (Table \ref{tab:datasets}), which are known \cite{LeskovecCKF05} to share many properties with real graphs.
\begin{observation}
As shown in Fig.~\ref{fig:scalability}, \method scales linearly with the number of edges in the graph.
\end{observation}
Notice that the algorithm can be trivially parallelized by finding the node affinity scores of the two graphs in parallel instead of sequential. Moreover, for each graph the computation of the similarity scores of the nodes to each of the $g$ groups can be parallelized. However, the runtime of our experiments refer to the sequential implementation.
%
%
%
\begin{figure}[th!]
\centering
   \includegraphics[width=0.6\columnwidth]{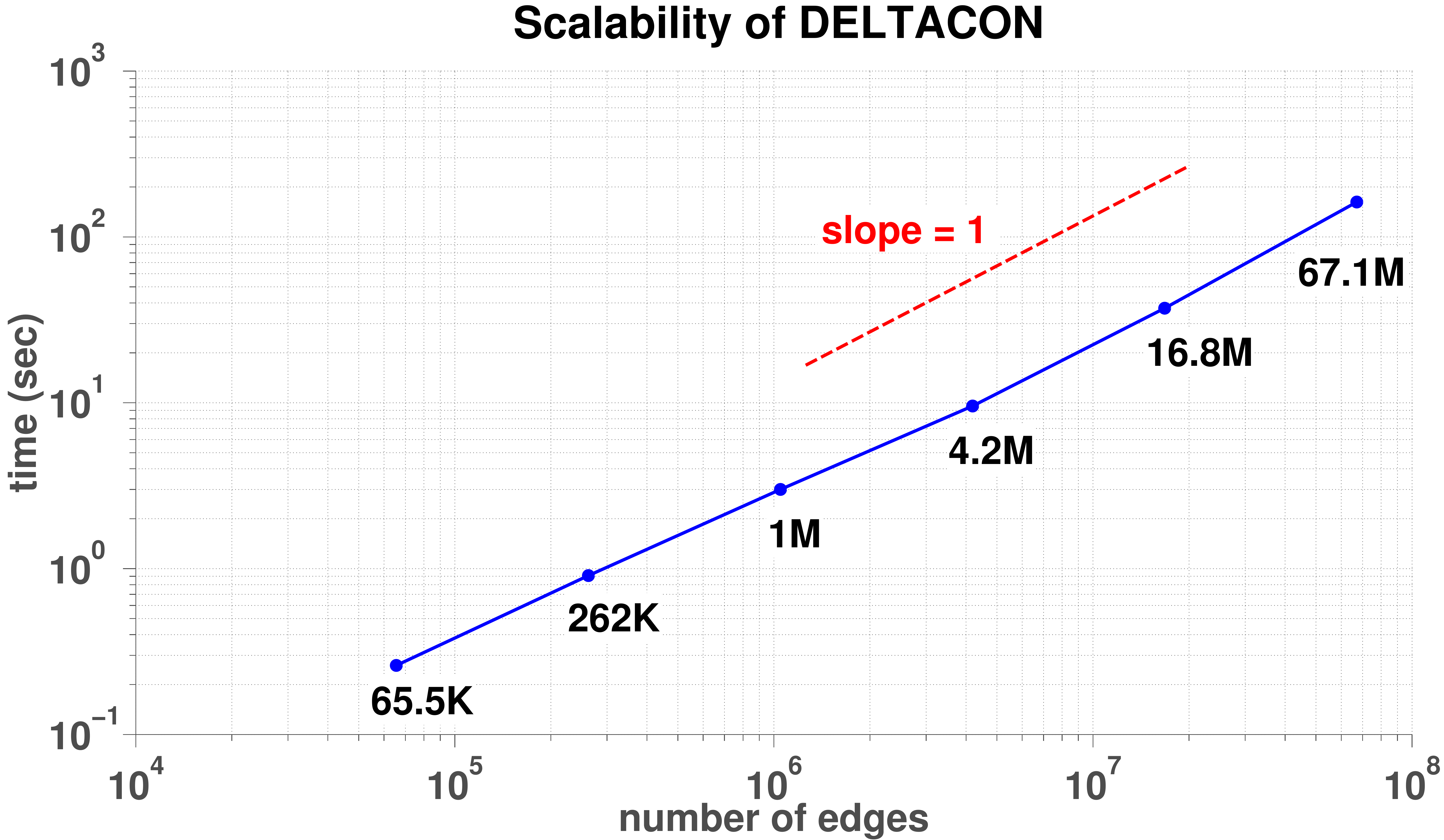}
\caption{\method is linear on the number of edges (time in sec. vs.~number of edges). The exact number of edges is annotated.}
\label{fig:scalability}
\end{figure}
\vspace{-0.3cm}
\section{\methodT at Work}
\label{sec:applications}
%
In this section we present two applications of graph similarity measures; we use \method and report our findings.
\begin{figure}
\centering
\includegraphics[width=0.65\columnwidth]{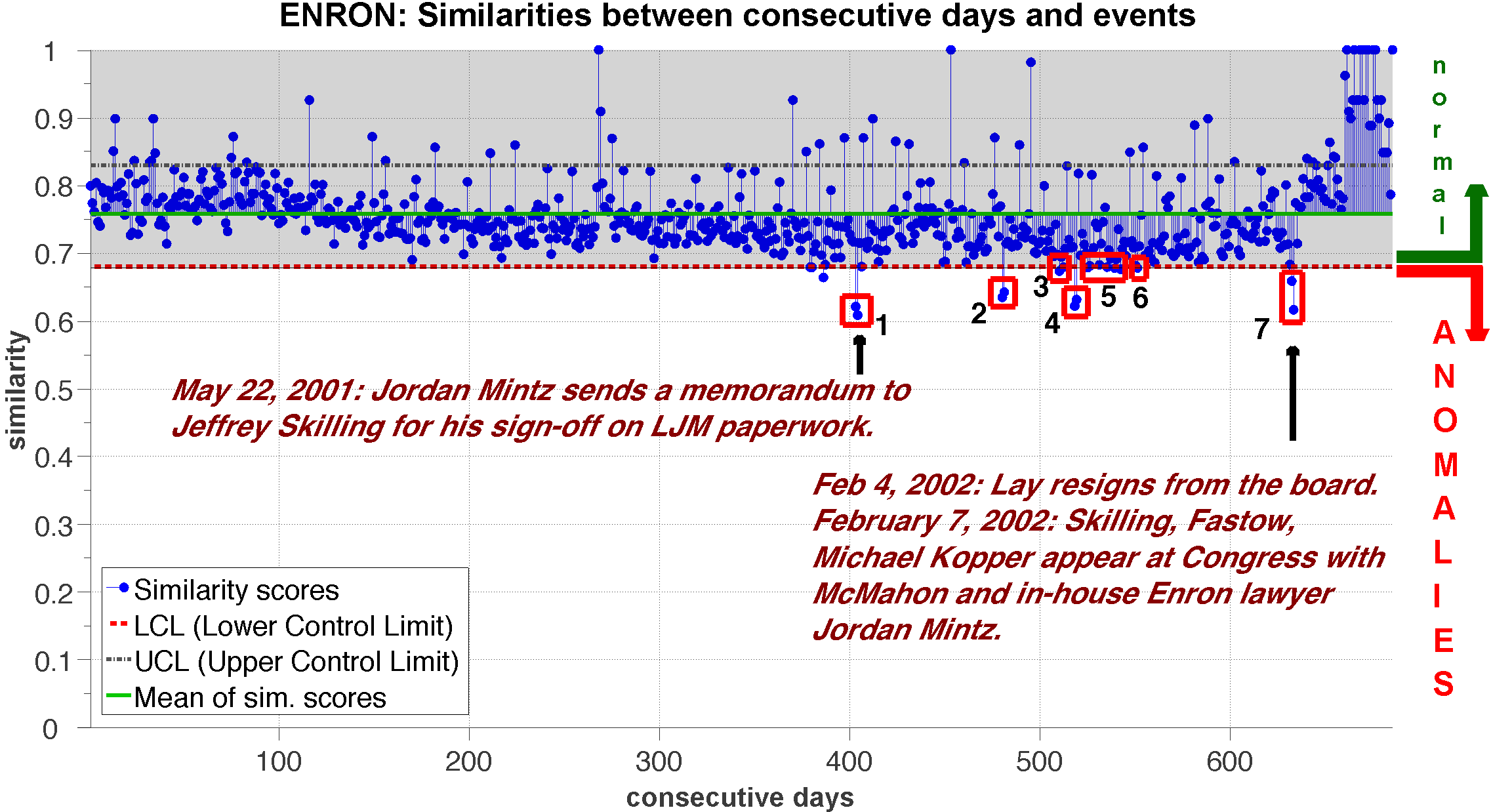}
\caption{Graph Anomaly Detection with \method. The marked days correspond to anomalies and coincide with major events in the history of Enron. The blue points are similarity scores between consecutive instances of the daily email activity between the employees, and the marked days are $3\sigma$ units away from the median similarity score. }
\label{fig:enron_timeline}
\vspace{-0.5cm}
\end{figure} 
\subsection{Enron.} First, we analyze the time-evolving ENRON graph. Figure \ref{fig:enron_timeline} depicts the similarity scores between consecutive daily who-emailed-whom graphs. By applying Quality Control with Individual Moving Range, 
we obtain the lower and upper limits of the in-control similarity scores. These limits correspond to median $\pm 3\sigma$ (The median is used instead of the mean, since appropriate hypothesis tests demonstrate that the data does not follow the normal distribution. Moving range mean is used to estimate $\sigma$.). 
Using this method we were able to define the threshold (lower control limit) below which the corresponding days are anomalous, i.e. they differ ``too much'' from the previous and following days. 
Note that all the anomalous days relate to crucial events in the company's history during 2001 (points marked with red boxes in Fig.~\ref{fig:enron_timeline}): 
(2) 8/21, Lay emails all employees stating he wants ``to restore investor confidence in Enron.'';
 (3) 9/26, Lay tells employees that the accounting practices are ``legal and totally appropriate'', and that the stock is ``an incredible bargain.''; 
 (4) 10/5, Just before Arthur Andersen hired Davis Polk \& Wardwell law firm to prepare a defense for the company; 
 (5) 10/24-25, Jeff McMahon takes over as CFO. Email to all employees states that all the pertinent documents should be preserved; 
 (6) 11/8,  Enron announces it overstated profits by 586 million dollars over 5 years.
%
%

Although high similarities between consecutive days do not consist anomalies, we found that mostly weekends expose high similarities. For instance, the first two points of 100\% similarity correspond to the weekend before Christmas in 2000 and a weekend in July, when only two employees sent emails to each other. It is noticeable that after February 2002 many consecutive days are very similar; this happens because, after the collapse of Enron, the email exchange activity was rather low and often between certain employees.
%
%

\subsection{Brain Connectivity Graph Clustering.} We also use \method for clustering and classification. For this purpose we study \emph{conectomes} -brain graphs-, which are obtained by Multimodal Magnetic Resonance Imaging \cite{GrayVLPV12}. 

In total we study the connectomes of 114 people; each consists of 70 cortical regions (nodes), and connections (weighted edges) between them. We ignore the strength of connections and derive an undirected, unweighted brain graph per person. In addition to the connectomes, we have attributes for each person (e.g., age, gender, IQ).

We first get the \method pairwise similarities between the brain graphs, and then perform hierarchical clustering using Ward's method (Fig.~\ref{fig:clustering}). As shown in the figure, there are two clearly separable groups of brain graphs. Applying t-test on the available attributes for the two groups created by the clusters, we found that the latter differ significantly (p-value=0.0057) in the Composite Creativity Index (CCI), which is related to the person's performance on a series of creativity tasks. Moreover, the two groups correspond to low and high openness index (p-value=0.0558), one of the ``Big Five Factors''; that is, the brain connectivity is different in people that are inventive and people that are consistent. Exploiting analysis of variance (ANOVA: generalization of t-test when more than 2 groups are analyzed), we tested whether the various clusters that we obtain from hierarchical clustering reflect the structural differences in the brain graphs. However, in the dataset we studied there is no sufficient statistical evidence that age, gender, IQ etc. are related to the brain connectivity.

\section{Related Work}
\label{sec:related}
%

\textbf{Graph Similarity.} The problems are divided in two main categories: 
\emph{(1) With Known Node Correspondence.}
Papadimitriou et al. \cite{PapadimitriouGM08} propose 5 similarity measures for directed web graphs. Among them the best is the Signature Similarity (SS), which is based on the SimHash algorithm, while the Vertex/Edge Overlap similarity (VEO) performs very well. Bunke \cite{Bunke06} presents techniques used to track sudden changes in communications networks for performance monitoring. The best approaches are the Graph Edit Distance and Maximum Common Subgraph. Both are NP-complete, but the former approach can be simplified given the application and it becomes linear on the number of nodes and edges in the graphs.
\emph{(2) With Unknown Node Correspondence.} Two approaches can be used: (a) feature extraction and similarity computation, (b) graph matching and application of techniques from the first category \cite{VogelsteinP11}, (c) graph kernels \cite{VishwanathanSKB10}. 
The research directions in this category include: $\lambda$-distance (\cite{Bunke06},\cite{Peabody03},\cite{WilsonZ08}), a spectral method that has been studied thoroughly; algebraic connectivity \cite{Fiedler73}; an SVM-based approach on global feature vectors \cite{LiSYZ11}; social networks similarity \cite{MacindoeR10}; computing edge curvatures under heat kernel embedding \cite{ElghawalbyH08}; comparison of the number of spanning trees \cite{Kelmans76}; fast random walk graph kernel \cite{KangTS12}.

Both research directions are important, but apply in different settings; if the node correspondence is available, the algorithms that make use of it can perform only better than methods that omit it. Here we tackle the former problem.

\textbf{Node affinity algorithms.} There are numerous node affinity algorithms; Pagerank \cite{BrinP98}, Personalized Random Walks with Restarts \cite{Haveliwala03}, the electric network analogy \cite{DoyleSL84}, SimRank \cite{JehW02}, and Belief Propagation \cite{YedidiaFW03} are only some examples of the most successful techniques. Here we focus on the latter method, and specifically a fast variation \cite{KoutraKKCPF11} which is also intuitive. All the techniques have been used successfully in many tasks, such as ranking, classification, malware and fraud detection (\cite{ChauNWWF11},\cite{McGlohonBASF09}), and recommendation systems \cite{KimES11}.

\reminder{Our method intrinsically takes into account the weights of the edges and penalizes more big and/or important changes in the graphs than
smaller ones, while being scalable.}

\reminder{Other requirements? -- same number of nodes, edges, result? Maybe create a table with this info}

\vspace{-0.1cm}
\section{Conclusions}
\label{sec:conclusions}
In this work, we tackle the problem of graph similarity 
when the node correspondence is known (e.g., similarity in
time-evolving phone
networks). 
Our contributions are:
\vspace{-0.2cm}
\begin{itemize*}
\item \emph{Axioms/Properties}: we formalize the problem of graph similarity by providing axioms, and desired properties.
\item \emph{Algorithm}: We propose \method, an algorithm that is (a) \emph{principled} (axioms $A1$-$A3$, in Sec.~\ref{sec:proposed}), (b) \emph{intuitive} (properties $P1$-$P4$, in Sec.~\ref{sec:experiments}), and (c) \emph{scalable}, needing on commodity hardware  \~{}160 seconds for a graph with over 67 million edges.
\item \emph{Experiments}: We evaluate the intuitiveness of \method, and compare it to 6 state-of-the-art measures.
\item \emph{Applications}: We use \method for temporal anomaly detection (ENRON), and clustering \& classification (brain graphs).
\end{itemize*}
Future work includes parallelizing our algorithm, as well as trying to partition the graphs in a more informative way (e.g., using elimination tree) than random.
%
%
%
%
\vspace{-0.1cm}
\section{Acknowledgements}
{\footnotesize The authors would like to thank Aaditya Ramdas, Aarti Singh, Elijah Mayfield, Gary Miller, and Jilles Vreeken for  their helpful comments and suggestions.

Funding was provided by the U.S. Army Research Office (ARO) and Defense Advanced Research Projects Agency (DARPA) under Contract Number W911NF-11-C-0088. 
Research was also sponsored by the Army Research Laboratory 
   and was accomplished under Cooperative Agreement Number W911NF-09-2-0053. It was also partially supported by an IBM Faculty Award. 
   The views and conclusions contained in this document are those 
   of the authors and should not be interpreted as representing 
   the official policies, either expressed or implied, 
   of the Army Research Laboratory or the U.S. Government. 
   The U.S. Government is authorized to reproduce and 
   distribute reprints for Government purposes notwithstanding 
   any copyright notation here on.
   }

\bibliographystyle{IEEEtran}
\vspace{-0.3cm}
\bibliography{BIB/abbreviations,BIB/references}  
%
%
\appendix
\section{Appendix}
\subsection{\textbf{From Fast Belief Propagation (\fabp) to \method}} 
 \label{sec:fabpToDeltacon}
 \fabp (\cite{KoutraKKCPF11}) is a fast approximation of the loopy \bp algorithm, which is guaranteed to converge and is described by the linear equation:
 $[\matI+a\matD-\cBP\matA]\bhalf= \phihalf,$ 
where $\phihalf$ is the vector of prior scores, $\bhalf$ is the vector of final scores (beliefs), $a=4\hhalf^{2}/(1-4\hhalf^{2})$, and $\cBP=2\hhalf/(1-4\hhalf^{2})$ are small constants, and $\hhalf$ is a constant that encodes the influence between neighboring nodes. 
By (a) using the MacLaurin expansion for $1/(1-4\hhalf^{2})$ and omitting the terms of power greater than 2, (b) setting $\phihalf=\unitv_i$ and $\bhalf=\bhalf_i$, and (c) setting $\hhalf=\epsilon/2$, we obtain eq.~(\ref{eq:fabp}), the core formula of \method.

\subsection{\textbf{Connection between \fabp and Personalized \rwr}}
\label{rwrfabp}

\begin{thm}
\label{thm:rwrfabp}
The \fabp equation (\ref{eq:fabp}) can be written in the Personalized RWR-like form:
$$[\matI-(1-\cBP')\matA_{*}\matD^{-1}]\bhalf_i= \cBP'\;\y,$$
where $\cBP'=1-\ccBP$, $\y=\matA_{*}\matD^{-1}\matA^{-1}\frac{1}{\cBP'}\unitv_i$ and $\matA_{*}=\matD (\matI+\aa\matD)^{-1}\matD^{-1}\matA\matD$.
\end{thm}

\begin{proof}
We begin from the derived \fabp equation (\ref{eq:fabp}) and do simple linear algebra operations:\\
$[\matI+\aa\matD-\ccBP\matA]\bhalf_i= \unitv_i$ \hfill ($\times$ $\matD^{-1}$ from the left) \\
$[\matD^{-1}+\aa\matI-\ccBP\matD^{-1}\matA]\bhalf_i= \matD^{-1}\unitv_i$ \hfill ($\matF=\matD^{-1}+\aa\matI$) \\
$[\matF-\ccBP\matD^{-1}\matA]\bhalf_i= \matD^{-1}\unitv_i$ \hfill ($\times$ $\matF^{-1}$ from the left) \\
$[\matI-\ccBP\matF^{-1}\matD^{-1}\matA]\bhalf_i= \matF^{-1}\matD^{-1}\unitv_i$ 
\hfill ( $\matA_{*} = \matF^{-1}\matD^{-1}\matA\matD$)\\
$[\matI-\ccBP\matA_{*}\matD^{-1}]\bhalf_i= (1-\ccBP)\;(\matA_{*}\matD^{-1}\matA^{-1}\frac{1}{1-\ccBP}\unitv_i)$ \hfill $\square$
\end{proof}


\subsection{\textbf{Proofs for Section \ref{sec:proposedDetails}}}
\label{sec:details}

\begin{lemma}
The time complexity of \method is linear on the number of edges in the graphs, i.e. O$(g \cdot max\{\m_{1}, \m_{2} \})$.
\end{lemma}
\begin{proof}[Proof of Lemma \ref{lem:complexity}]
By using the Power Method \cite{KoutraKKCPF11}, the complexity of solving eq.~(\ref{eq:fabp}) is O$(|\edgeSet_{i}|)$ for each graph $(i=1,2)$. The node partitioning needs O$(\n)$ time; the affinity algorithm is run $g$ times in each graph, and the similarity score is computed in O$(g\n)$ time.
Therefore, the complexity of \method is 
O$((g+1)\n + g(\m_{1}+\m_{2}))$, where $g$ is a small constant. 
Unless the graphs are trees, $|\edgeSet_{i}|<\n$, so the complexity of the algorithm 
reduces to O$(g(\m_{1}+\m_{2}))$. Assuming that the affinity algorithm is run on the graphs in parallel, since there is no dependency between the computations, \method has 
complexity O$(g \cdot max\{\m_{1}, \m_{2} \}). \hfill \square$
%
\end{proof}
\begin{lemma}
\label{lem:lincomb}
The affinity score of each node to a group (computed by \method) is equal to the sum of the affinity scores of the node to each one of the nodes in the group individually (computed by \methodSlow).
\end{lemma}
\begin{proof}
Let ${\mathbf B}=\matI+\epsilon^{2}\matD-\epsilon\matA$. Then \methodSlow consists of solving for every node $i \in \nodeSet$ the equation ${\mathbf B}\bhalf _{i}= \unitv_{i}$;
\method solves the equation ${\mathbf B}\bhalf_k' = \phihalf_{k}$ for all groups $k \in (0,g]$, where $\phihalf_{k} = \sum_{i \in group_{k}}\unitv_{i}$. Because of the linearity of matrix additions, it holds true that $\bhalf_{k}' = \sum_{i \in group_{k}}\bhalf_{i}$, for all groups $k$. \hfill $\square$
\end{proof}

\begin{thm}
\method's similarity score between any two graphs $G_{1}$, $G_{2}$ upper bounds the actual \methodSlow's similarity score, i.e. $sim_{DC-0}(G_{1}, G_{2}) \le sim_{DC}(G_{1}, G_{2})$.
\end{thm}
\begin{proof}[Proof of Theorem \ref{thm:upperBound}]
Let $\matB_{1}, \matB_{2}$ be the $\n \times \n$ final-scores matrices of $G_{1}$ and $G_{2}$ by applying \methodSlow, and $\matB'_{1}, \matB'_{2}$ be the respective $\n \times g$ final-scores matrices by applying \method. We want to show that \methodSlow's distance
\begin{center}$d_{DC_0} = \sqrt{\sum_{i=1}^{\n}\sum_{j=1}^{\n}(\sqrt{s_{1,ij}}-\sqrt{s_{2,ij}})^{2}}$ \end{center}
is greater than \method's distance
\begin{center}$d_{DC} = \sqrt{\sum_{k=1}^{g}{\sum_{i=1}^{\n}{(\sqrt{s'_{1,ik}}-\sqrt{s'_{2,ik}})^{2}} }}$\end{center}
or, equivalently, that $d_{DC_0}^{2} > d_{DC}^{2}$. It is sufficient to show that for one group of \method, the corresponding summands in $d_{DC}$ are smaller than the summands in $d_{DC_0}$ that are related to the nodes that belong to the group. By extracting the terms in the squared distances that refer to one group of \method and its member nodes in \methodSlow, and by applying Lemma~\ref{lem:lincomb}, we obtain the following terms:
\begin{center}$t_{DC_0} = \sum_{i=1}^{\n}\sum_{j \in group}(\sqrt{s_{1,ij}}-\sqrt{s_{2,ij}})^{2}$\end{center}
\begin{center}$t_{DC} = \sum_{i=1}^{\n}{(\sqrt{\sum_{j \in group}{s_{1,ij}}}-\sqrt{\sum_{j \in group}s_{2,ij}})^{2}}.$\end{center}
Next we concentrate again on a selection of summands (e.g. $i=1$), we expand the squares and use the Cauchy-Schwartz inequality to show that 
\begin{center}$\sum_{j \in group}{\sqrt{s_{1,ij}s_{2,ij}}} < \sqrt{\sum_{j \in group}s_{1,ij}\sum_{j \in group}s_{2,ij}},$\end{center}
or equivalently that $t_{DC_0} > t_{DC}$. \hfill $\square$ 
\end{proof}

\subsection{\textbf{Satisfaction of the Axioms}} 
\label{sec:axioms}
Here we elaborate on the satisfiability of the axioms by \methodSlow and \method.

A1.\emph{ Identity Property}: $sim(G_1, G_1) = 1$. 

The proof is straightforward; the affinity scores are identical for both graphs.

A2.\emph{ Symmetric Property}: $sim(G_1, G_2) = sim(G_2, G_1)$.

The proof is straightforward for \methodSlow.
For the randomized algorithm, \method, it can be shown that $sim(G_1, G_2) =  sim(G_2, G_1)$ on average.

A3.\emph{ Zero Property}: $sim(G_1, G_2)  \to 0$ for $\n \to \infty$, where $G_1$ is the clique graph ($K_{\n}$), and $G_2$ is the empty graph (i.e., the edge sets are complementary). 

We restrict ourselves to a sketch of proof, since the proof is rather intricate.
\begin{proof}[(Sketch of Proof - Zero Property)]
First we show that all the nodes in a clique get final scores in $\{s_{g}, s_{ng}\}$, depending on whether they are included in group $g$ or not. Then, it can be demonstrated that the scores have finite limits, and specifically $\{s_{g}, s_{ng}\} \to \{\frac{\n}{2g}+1, \frac{\n}{2g}\}$ as $\n \to \infty$ (for finite $\frac{\n}{g})$. Given this condition, it can be directly derived that the \matusita between the $\matB$ 
matrices of the empty graph and the clique becomes arbitrarily large. So, $sim(G_1, G_2) \to 0$ for $\n \to \infty$. \hfill $\square$
\end{proof}

\subsection{\textbf{Satisfaction of the Properties}} 
\label{sec:properties}
Here we give some theoretical guarantees for the most important property, ``edge importance'' (P1). We prove the satisfiability of the property in a special case; generalizing this proof, as well as the proofs of the rest properties is theoretically interesting and remains future work.  

Special Case [Barbell graph]: Assume $A$ is a barbell graph with $n_1$ and $n_2$ nodes in each clique (e.g., $B10$ with $n_1=n_2=5$ in Fig.~\ref{fig:BPtests}), $B$ does not have one edge in the first clique, and $C$ does not have the ``bridge'' edge.
\begin{proof}
From eq. (\ref{eq:fabp}), by using the Power method we obtain the solution:
$$\bhalf_i=[\matI+(\epsilon \matA - \epsilon^2 \matD) + (\epsilon \matA - \epsilon^2 \matD)^2 + ...] \unitv_i \Rightarrow$$
$$\bhalf_i \approx [\matI+ \epsilon \matA + \epsilon^2 (\matA^2 - \matD)] \unitv_i,$$
where we ignore the terms of greater than second power.
By writing out the elements of the $\matB_A, \matB_B, \matB_C$ matrices as computed by \methodSlow and the above formula, and also the \matusita between graphs A, B and A, C, we obtain the following formula for their relevant difference:
$$d(A,C)^2-d(A,B)^2 = 2\epsilon\{\epsilon(n-f) +1 - (\frac{2\epsilon^3(n_1-2)}{c_1^2} + \frac{\epsilon}{c_2^2})\},$$
where $c_1=\sqrt{\epsilon+\epsilon^2(n_1-3)}+\sqrt{\epsilon+\epsilon^2(n_1-2)}$ and $c_2=\sqrt{\epsilon^2(n_1-2)}+\sqrt{\epsilon+\epsilon^2(n_1-2)}.$
We can show that $c_1 \geq 2 \sqrt{\epsilon}$ for $\hhalf < 1$ (which holds always) and $c_2 \geq \sqrt{\epsilon},$
where $f=3$ if the missing edge in graph $B$ is adjacent to the ``bridge'' node, and $f=2$ in any other case. So, $d(A,C)^2-d(A,B)^2 \geq 0$.

This property is not always satisfied by the euclidean distance. \hfill $\square$
\end{proof}

\end{document}